\documentclass[a4paper,runningheads,11pt,UKenglish]{article}
\usepackage{amsthm,amssymb}
\usepackage[latin1]{inputenc}
\usepackage{amsmath}
\usepackage{amsfonts}
\usepackage{amssymb}
\usepackage{makeidx}
\usepackage{graphicx}
\usepackage[left=3.00cm, right=3.00cm, top=3.00cm, bottom=3.00cm]{geometry}

\newtheorem{theorem}{Theorem}
\newtheorem{lemma}{Lemma}
\newtheorem{remark}{Remark}

\title{Longest Increasing Subsequence\\ under Persistent Comparison Errors}

\author{Barbara Geissmann \\ Department of Computer Science\\ ETH Zurich, Zurich, Switzerland\\ {barbara.geissmann@inf.ethz.ch}}

\usepackage{amsmath}
\DeclareMathOperator{\rank}{rank}
\DeclareMathOperator{\pos}{pos}
\DeclareMathOperator{\disl}{disl}
\newcommand{\APXSORT}{S^{apx}}
\newcommand{\SORT}{S^{sort}}

\usepackage{thmtools}
\usepackage{thm-restate}
\usepackage{xcolor}

\usepackage[ruled,vlined,linesnumbered]{algorithm2e}

\begin{document}
	
	\maketitle

\begin{abstract}
	We study the problem of computing a \textit{longest increasing subsequence} in a sequence $S$ of $n$ distinct elements in the presence of \textit{persistent} comparison errors.
	In this model,\footnote{Braverman and Mossel, \textit{Noisy sorting without resampling}, SODA, 2008} every comparison between two elements can return the wrong result with some fixed (small) probability $ p $, and comparisons cannot be repeated.
	Computing the longest increasing subsequence exactly is impossible in this model, therefore, the objective is to identify a subsequence that (i) is indeed increasing and (ii) has a length that approximates the length of the longest increasing subsequence.

We present asymptotically tight upper and lower bounds on both the approximation factor and the running time.
In particular, we present an algorithm that computes an $O(\log n)$-approximation in time $O(n\log n)$, with high probability. This approximation relies on 
the fact that that we can approximately sort\footnote{Geissmann, Leucci, Liu, and Penna, \textit{Optimal Sorting with Persistent Comparison Errors}, ArXiv e-prints 1804.07575, 2018} $n$ elements in $O(n\log n)$ time such that the maximum dislocation of an element is at most $O(\log n)$.
For the lower bounds, we prove that (i) there is a set of sequences, such that on a sequence picked randomly from this set every algorithm must return an $\Omega(\log n)$-approxi\-mation with high probability, and (ii) any $O(\log n)$-approximation algorithm for longest increasing subsequence requires $\Omega(n \log n)$ comparisons, even in the absence of errors.
\end{abstract}

\section{Introduction}
When dealing with complex systems and large volumes of information, it is often the case that at least part of the involved data will be inconsistent.
These inconsistencies can be \emph{intrinsic}, i.e., they might shed from the fact that the data is obtained from an inherently \emph{noisy} source (this is typically the case in human-produced data),
or they might be the result of corruptions caused by random errors (think, for instance, of random memory faults or communication errors).
It is therefore important to understand how the classical techniques used to solve basic algorithmic problems can cope with such errors.

In this paper, we consider the problem of computing a \emph{longest increasing subsequence} $LIS(S)$ in a given sequence $S$ of distinct elements 
--a fundamental task that appears naturally in many areas, such as in probability theory and combinatorics \cite{aldous1999longest,baik1999distribution}, scheduling \cite{0305-4470-39-29-L01,potts1991permutation}, 
 and computational biology \cite{delcher,Zhang03}-- in  presence of \emph{random persistent comparison errors}. 
 
 In this model, every comparison between two elements is wrong with some small fixed probability $p$, and correct with probability $1-p$. The comparison results are independent over all pairs of elements, and comparisons cannot be repeated. Note that this is equivalent to say that repeating the same comparison multiple times yields each time the same result. Hence, comparison results are persistent: always wrong or always correct. Furthermore, we assume that we cannot inspect the values of the elements, but only use such element comparisons.
Because of these comparison errors, it is impossible to compute $ LIS(S) $ correctly, instead, we seek to return a sequence that (i) is indeed increasing and that (ii) has some guaranteed minimum length depending on the length of the longest increasing sequence $l:=|LIS(S)|$. In particular, we are interested in algorithms that return an increasing sequence of length at least  $\frac{1}{r}\cdot l$, where $r$ is the \emph{approximation factor}.

This error model has been first employed by Braverman and Mossel \cite{BravermanM08}, who studied the problem of sorting. Other work on sorting followed (see \cite{ISAAC,STACS,KleinPSW11}) and the model has been studied also for finding the minimum, searching, and linear programming in two dimensions \cite{KleinPSW11}.
In this paper, we will present an algorithm that returns an $O(\log n)$-approximation on the longest increasing subsequence in $O(n\log n)$ time, with high probability. 
Moreover, we will prove that
this approximation factor is the best possible as $\Omega(\log n)$ is also a lower bound,
regardless of the running time, and that any $(\log n)$-approximation algorithm requires $\Omega(n \log n)$ comparisons, even in the absence of comparison errors.

\subsection{Related Work}\label{sec:related-work}
There are several algorithms to compute a longest increasing subsequence of a sequence $S$, if no comparison errors happen.
Typically, they are based on a common underlying algorithmic idea:
They process the elements one by one and maintain for each length found so far the increasing subsequence of this length that ends with the smallest possible element seen so far. We shall call this algorithmic idea the \emph{Core-Algorithm} to compute a longest increasing subsequence. 
The running time of the Core-Algorithm is $O(n\log n)$ in the decision-tree model (see for instance \cite{BespamyatnikhS00a,ChandramouliG14,Fredman75}). This time complexity is tight, as shown in \cite{Fredman75}.
In the RAM model, where one can also inspect the values, the algorithm can be implemented to run in $O(n\log\log n)$ time \cite{CrochemoreP10,YangHC05}.
All the results can be parameterized to $O(n\log l)$ or $O(n \log\log l)$, respectively, where $l$ is the length of the longest increasing  subsequence. 

The longest increasing subsequence of $S$ is also the \emph{longest common subsequence} between $S$ and the sorted sequence of the elements in $S$. This implies an $O(n^2)$ time (or $O(n^2/\log n)$ time if optimized) algorithm to find the longest increasing subsequence when using the standard dynamic programming technique that is used to find longest common subsequences \cite{Fredman75,MASEK198018}.

The model with random persistent comparison errors has been extensively studied for finding the smallest element, for searching, and for sorting (see for instance  \cite{BravermanM08,ISAAC,STACS,KleinPSW11}). 
A common way to measure the quality of an output sequence in terms of sortedness, is to consider the \textit{dislocation} of the elements. The dislocation of an element is the absolute difference between its position in the output sequence and its position in the correctly sorted sequence (its rank). Typically, one considers the \textit{maximum dislocation} of any element in the output sequence and the \textit{total dislocation} (the sum of the dislocations of all elements).
It has been shown for instance in \cite{OPT-SORTING}, that there is an algorithm with running time $O(n\log n)$ which achieves simultaneously maximum dislocation $O(\log n)$ and total dislocation $O(n)$ with high probability, and that this is indeed the best one can hope for (i.e., there exist matching lower bounds that show that no possibly randomized algorithm can sort such that, with high probability, the maximum dislocation is $o(\log n)$ or the total dislocation is $o(n)$).
A maximum dislocation of $O(\log n)$ implies the following: on the positive side, it is possible to derive the correct relative order of two elements whose ranks differ by at least $\Omega(\log n)$; on the negative side, this is not possible for two elements whose ranks differ by less than $O(\log n)$.
The results on the maximum dislocation of sorting are of interest for the problem of finding the longest  increasing subsequence, because an increasing subsequence is also a sorted subsequence.

\subsection{Our Contribution}
We prove asymptotically tight upper and lower bounds on both the approximation factor and the running time for longest increasing subsequence under persistent comparison errors.
For the upper bounds, we define an \emph{Approximation-Algorithm} that computes an $O(\log n)$-approxi\-mation to the longest increasing subsequence of $S$.
In fact, it even finds the longest possible increasing subsequence under the implication that we cannot sort better than obtaining an order with maximum dislocation $O(\log n)$.
Formally, we prove the following result:
\begin{theorem}[Upper Bounds]\label{thm:upper}
	 For any sequence $S$ that contains $n$ distinct elements, our {Approximation-Algorithm} computes an $O(\log n)$-approximation to the longest increasing sequence of $S$, in $O(n\log n)$ time, with probability at least $1-\frac{1}{n}$.
\end{theorem}

\noindent
This result on the upper bound can be generalized to other error models. In fact, if we are given or able to obtain an approximately sorted sequence with maximum dislocation $d$, then our Approximation-Algorithm will return a $2d$-approximation to the longest increasing subsequence. We discuss this point in the Conclusion (Section~\ref{sec:conclusion}).

To prove our lower bound on the approximation factor of any algorithm solving $LIS(S)$ under persistent comparison errors with high probability, we will identify  a small collection of sequences that contain a longest increasing sequence of size $\Theta(\log n)$ and that are likely to look the same in our error model. Then, we show for any algorithm that if it \emph{succeeds} on one sequence of this collection by returning a constant number of elements of this increasing  sequence it must fail on another sequence. 
In particular, we will prove the following theorem:
\begin{theorem}[Lower Bound -- Approximation Factor]\label{thm:lower}
	There exists a collection of sequences $\mathcal{S}$ (permutations of length $n$) and a probability distribution on $\mathcal{S}$, such that no algorithm can return an $O(\log n)$-approximation (for s suitable hidden constant that depends on $p$) of the longest increasing subsequence with probability $1-\frac{1}{n}$.
\end{theorem}

We prove a lower bound of $\Omega(n\log n)$ on the number of comparisons (which is a lower bound on the running time)  needed to compute an $O(\log n)$-approximation by considering the easier case in which all comparisons are correct, and by adapting the techniques used in \cite{Fredman75} for proving a similar lower bound for exact (i.e., 1-approximate) algorithms:
\begin{theorem}[Lower Bound -- Running Time]\label{thm:lower-time}
	Any $(\log n)$-approximation algorithm for longest increasing subsequence requires $\Omega(n \log n)$ comparisons, even if no errors occur.
\end{theorem}

\section{Preliminaries}
Since we assume that all elements in the input sequence $S=\langle s_1,s_2,\dots,s_n \rangle$ are distinct, we can also assume, for easier analysis and readability, that $S$ is a permutation of the numbers (elements) $\{1,\dots,n\}$. 
By our error model, the elements in $S$ posses a true \mbox{linear order} , i.e., $\SORT := \langle 1,\dots,n\rangle$, however, this order can only be observed through erroneous  comparisons.

For two distinct elements $x$ and $y$, we will write $x<y$ to denote that $x$ is smaller than $y$ according to the true linear order (resp. $x>y$ to denote that $x$ is larger than $y$ according to the true linear order), and we will write $x\prec y$ (resp. $x\succ y$) to mean that $x$ is observed to be smaller (resp. larger) than $y$ in the comparison result.
For a given sequence $S$ and an element $x\in S$, we define $\rank(x,S) = 1+|\{y\in S\colon y<x\}|$ to be the \emph{true rank} of element $x$ in $S$ (note that ranks start from 1), and we define $\pos(x,S)\in [1,|S|]$ to be the \emph{position} of $x$ in $S$ (positions also start \mbox{from 1}). The dislocation of $x$ in $S$ is then $\disl(x,S) = |pos(x,S)-\rank(x,S)|$, and the \emph{maximum dislocation} of $S$  is $\disl(S) = \max_{x \in S} \disl(x,S)$.
%
For a given sequence $S$, we let $C\in \{\prec,\succ \}^{\binom{n}{2}}$ denote the comparison outcomes that we can observe. For $C=\langle c_1,\dots, c_{\binom{n}{2}} \rangle $, this means that if $c_k=c_{(i-1)n+j}=~``\prec"$ with $1\le i<n$ and  $i<j\le n$, then $s_i \prec s_j $ (resp. $s_i \succ s_j $ if $c_k=c_{(i-1)n+j}=``\succ"$).
Finally, for $z \in \mathbb{R}$, we write $\log z$ 
for the binary logarithm of~$z$.

We continue the preliminaries with some results on sorting that we will use 
to prove our upper bound on the approximation factor.
\begin{theorem}[Theorem 3 in \cite{OPT-SORTING}]\label{thm:sort}
	There is an algorithm that approximately sorts, in $O(n\log n)$ worst-case time, $n$ elements subject to random persistent comparison errors so that the maximum  dislocation of the resulting sequence is $O(\log n)$, with probability $1-\frac{1}{n}$.
\end{theorem}

\begin{lemma}\label{lem:relative-order}
Let $ \APXSORT=\langle apx_1,apx_2,\dots,apx_n \rangle $.	If $\disl(\APXSORT)\le d$, then for $1\le i<n-2d $, $apx_i$ and $apx_{i+2d}$ are in correct relative order: 
$	\pos(apx_i,\SORT) < \pos(apx_{i+2d},\SORT)\, .$
\end{lemma}
\begin{proof}
	Since the maximum dislocation in $ \APXSORT $ is at most $d$,  $\pos(apx_i,\SORT) \in \{i-d,\dots,i+d\} $ and $\pos(apx_{i+2d},\SORT) \in \{i+d,\dots,i+3d\} $. These intervals intersect in at most one position, and the claim follows since no two elements can appear in the same position.
\end{proof}

\section{Upper Bound and Approximation-Algorithm}
We will modify the so-called Core-Algorithm (as named in Section~\ref{sec:related-work}, Related work) that computes a longest increasing subsequence in the absence of comparison errors, such that it computes an $O(\log n)$-approximation with high probability in our error model.
Before we do so, we first
show that it is possible to identify a $2d$-approximation by looking at $S$ and a sequence $\APXSORT$ with maximum dislocation $d$. 
Since we can sort such that the maximum dislocation is $O(\log n)$ (see Theorem~\ref{thm:sort}), this implies an $O(\log n)$-approximation on $ LIS(S) $.

\subsection{Upper Bound}\label{sec:upperbound}
The proof of the upper bound is based on the following fact and observation:
\begin{itemize}
	\item 	Without any comparison errors, the problem of finding $ LIS(S) $ is equivalent to the problem of finding a \textit{longest common subsequence} between $S$ and $\SORT$, where $\SORT$ is the correctly sorted order of the elements in $S$.
	\item This leads to the following observation. Let $\APXSORT$ be the sequence obtained from approximately sorting $S$ with comparison errors and consider now $\APXSORT$ as the total order over all elements, i.e., for each pair of elements, their comparison result is \textit{redefined} as their relative order in $\APXSORT$.
	Furthermore, let $A$ be any algorithm that solves $LIS(S)$ in the absence of errors. 
	If $A$ uses the redefined comparison results, it computes the \textit{longest common subsequence} $LCS(S,\APXSORT)$ between $S$ and $\APXSORT$. 
\end{itemize}

\noindent
The immediate idea of computing $LCS(S,\APXSORT)$ comprises some difficulties, since this subsequence is not necessarily increasing and, on top of that, $|LCS(S,\APXSORT)|$ might be smaller than $|LIS(S)|$. However, we can still get a first approximation.
Assume that $\APXSORT$ has maximum dislocation at most $d$. Lemma~\ref{lem:relative-order} implies that we obtain an increasing subsequence when taking every $2d$-th element of $LCS(S,\APXSORT) $. And the maximum dislocation implies that the elements in the subset containing every $2d$-th element of $LIS(S)$ appear in the same relative order in  $\APXSORT$, thus $|LCS(S,\APXSORT)| \ge \frac{1}{2d}|LIS(S)|\, .$ When put together, we get a $4d^2$-approximation. 

This approximation factor can be improved, and it turns out that considering common subsequences whose elements lie (at least) $2d$ positions apart in $\APXSORT$ is actually a good start: By Lemma~\ref{lem:relative-order}, a common subsequence 
between $S$ and $\APXSORT$ is  \textit{increasing} if for every pair of adjacent elements 
in this subsequence
their positions in $\APXSORT$ differ by at least $2d$. 
Therefore, we say that a sequence $S'=(s_1',s_2',\dots,s_m')$ is \emph{$2d$-distant} in $\APXSORT$ if\begin{equation}\label{eq:increasing-condition}
\pos(s_i',\APXSORT)+2d\le\pos(s_{i+1}',\APXSORT)\, ~~~~\text{ for } 1\le i < m\, .
\end{equation}
Notice that any (increasing) subsequence of $S$ that is $2d$-distant in $\APXSORT$ is automatically also a common (increasing) subsequence of $S$ and $\APXSORT$.
This observation suggests the following easy recipe to obtain a $2d$-approximation on longest increasing subsequence:
\begin{itemize}
\item First, partition the elements into $2d$ subsets, such that every $2d$-th element in $\APXSORT$ gets into the same subset, and obtain $2d$ input subsequences based on this partition.
\item Then, on every input subsequence, run any algorithm that computes a longest increasing subsequence if no comparison errors happen, and return the longest result.
\end{itemize}
By pigeon hole principle and since every input subsequence is now $2d$-distant in $\APXSORT$, the longest result must be a $2d$-approximation on $|LIS(S)|$. This recipe however is not optimal in the sense that in many cases, we could do better and find a longer subsequence in $S$ that is still $2d$-distant in $\APXSORT$. In fact, we lose up to a factor $2d$ in the case where $LIS(S)$ is already $2d$-distant in $\APXSORT$, but these elements are equally distributed among all input subsequences.
For this reason, we will define an approximation algorithm that finds the longest increasing subsequence in $S$ that is $2d$-distant in $\APXSORT$. We conclude this section with the obvious lemma.

\begin{lemma}\label{lem:apx-factor}
	The longest subsequence $S^*$ of $S$  that is $2d$-distant in  $\APXSORT$ has length at least\[|S^*|\ge\frac{1}{2d}|LIS(S)|\, .\]
\end{lemma}

\subsection{Approximation-Algorithm}\label{sec:algorithm}
Consider the \emph{Core-Algorithm} described in Algorithm~\ref{alg:core} that computes the longest increasing subsequence of the input sequence $S$ in the error-free case.
The algorithm processes the input elements one by one, maintaining the longest increasing subsequence found so far. 
In particular, it maintains a parameter $k$ and an array $L$, such that $k$ is the length of the longest increasing subsequence found so far and $L$ contains an entry for each length 1 to $k$, such that $L[i]$ stores the smallest element processed so far that can be at the end of an increasing subsequence of \mbox{length $i$.} 
\begin{itemize}
	\item The first element is placed to $L[1]$ and $k$ is set to 1.
	\item Each subsequent element $x$ is placed to $L[j+1]$, such that $j$ is the largest position where $y=L[j]$ is smaller than $x$.
	\item If $x$ is placed to $L[k+1]$, then $k$ is updated to $k+1$.
	\item Whenever a new element $x$ is placed, put a pointer $prec$ from $x$ to the element in $y=L[j]$, that, by construction, has a lower value than $x$.
	\item In the end, follow these pointers from the top element of the last pile to recover the longest increasing subsequence (in reverse order). 
\end{itemize}
An entry $L[j]$ basically represents the increasing sequence of length $j$ that ends with the smallest possible element processed so far. When an element $x$ is inserted into some position $L[j+1]$ this means that it is appended to the sequence represented by $L[j]$. Hence, $x$ either increases the longest increasing sequence so far (case $j=k$) or the sequence $L[j+1]$ gets replaced by this new sequence (case $x<L[j+1]$).
\begin{algorithm}[t]\caption{$Core$-$Algorithm(S=\langle s_1,\dots,s_n \rangle)$}\label{alg:core}\small
	$L[1] \longleftarrow s_1$;
	$k\longleftarrow 1$\;
	\ForEach{$ i=2,\dots,n $}{$ x\longleftarrow s_i $\;
		\lIf{$x<L[1]$}{$L[1]\longleftarrow x$}
		\Else{
			$j\longleftarrow \max\{j\le k \colon L[j]< x \}$\;
			\lIf{$j=k$}{$k\longleftarrow k+1$}
			$L[j+1]\longleftarrow x$;
			$prec[x]\longleftarrow L[j]$\;}}
	$lis[1] \longleftarrow L[k]$\;
	\ForEach{$i=2,\dots,k$}{
		$lis[i]\longleftarrow prec[lis[i-1]]$}
	\Return $lis$\;
\end{algorithm}

Our \emph{Approximation-Algorithm}, as described in Algorithm~\ref{alg:apx}, is obtained by modifying the Core-Algorithm such that it works in our error model.  
\begin{itemize}
	
	\item We first approximately sort (using the algorithm from \cite{OPT-SORTING}, see also Theorem~\ref{thm:sort} in the current paper) the elements of $S$ to obtain $\APXSORT$, and we redefine the comparison outcomes based on this total order, i.e., the result of a comparison between two elements now corresponds to their relative order in $\APXSORT$.

\item  To compute a suitable subsequence, we change the algorithm so that it remembers the longest $2d$-distant in $\APXSORT$ subsequences instead of the longest increasing subsequences. This implies that an element $x$ is only appended to an (intermediate) subsequence that ends with element $y$, if $\pos(y,\APXSORT)+2d < \pos(x,\APXSORT)$.
\end{itemize}

For easier analysis, we introduce some additional notation. We
call one execution of the lines \ref{itstart} to \ref{itend} of Algorithm~\ref{alg:apx} an \emph{iteration}, and enumerate them such that element $s_i$ is considered in iteration $i$. We also say that line \ref{itone}  corresponds to the first iteration.
Furthermore, we denote by $L_t$ and $k_t$ the state and the value of $L$ and $k$ after the $t$-th iteration, respectively, and for any $j\le k_t$, we call the subsequence $\langle L_t[j], prec[ L_t[j]], prec[prec[ L_t[j]]],\dots \rangle $ with length $j$ the \emph{implied} sequence of $L_t[j]$.

\begin{algorithm}[t]\caption{$Approximation$-$Algorithm(S=\langle s_1,\dots,s_n \rangle)$}\label{alg:apx}\small
	
	$\APXSORT \longleftarrow \text{approximately sort $S$ as shown in \cite{OPT-SORTING}}$~\;
	$d\longleftarrow c\cdot\log n$ ~~~~~~~~~~~~~~~~{\color{darkgray}// $\exists c$ s.t. w.h.p. $\disl(s)\le c\cdot\log n$ \cite{OPT-SORTING}}\;
	$L[1] \longleftarrow s_1$;
	$k\longleftarrow 1$\label{itone}\;
	\ForEach{$ i=2,\dots,n $\label{itstart}}{$ x\longleftarrow s_i $\;
		\lIf{$\pos(x,\APXSORT)<\pos(L[1],\APXSORT)$\label{firstcomp}}{$L[1]\longleftarrow x$}
		\Else{
		 $j\longleftarrow \max\{j\le k \colon \pos(L[j],\APXSORT)< \pos(x,\APXSORT) \}$\label{secondcomp}\;
		\If{$\pos(L[j],\APXSORT)+2d\le \pos(x,\APXSORT)$\label{twoddist}}{
		\lIf{$j=k$}{$k\longleftarrow k+1$}
		$L[j+1]\longleftarrow x$;
		$prec[x]\longleftarrow L[j]$\label{itend}\;}}}
	$lis[1] \longleftarrow L[k]$\;
	\ForEach{$i=2,\dots,k$}{
		$lis[i]\longleftarrow prec[lis[i-1]]$\;}
	\Return $lis$\;
\end{algorithm}

\begin{lemma}\label{lem:alg1}
	For every $t\le n$, after the $t$-th iteration of our Approximation-Algorithm,
	every implied sequence is a subsequence of $S$ that is $2d$-distant in $\APXSORT$. Moreover, $\langle L_t[1],\dots,L_t[k_t] \rangle$ is also $2d$-distant in $\APXSORT$.
\end{lemma}
\begin{proof}
	For any $t$ and $j\le k_t$, let $S'=\langle s'_1,\dots,s'_m \rangle$ be the implied sequence of $L_t[j]$. Observe that to every element $s'_i\in S'$, such that $i>1$, the algorithm has assigned $s_{i-1}'$ as its predecessor. Since the predecessor of any element can only have been processed in an earlier iteration, $S'$ is a subsequence of $S$.
	
	It follows by induction, that the condition on line \ref{twoddist} in Algorithm~\ref{alg:apx} ensures that $S'$ is $2d$-distant in $\APXSORT$:
	It is trivial to see for $t=1$, thus, assume that every implied sequence before the $t$-th iteration is $2d$-distant in $\APXSORT$. 
	If $s_t$ is inserted into $L[j]$ (nothing changes in the other case), the implied sequence of $L_t[j]$ is equal to $s_t$ appended to the implied sequence of $L_{t-1}[j-1]$ (if it exists). By hypothesis and the condition on line \ref{twoddist}, $L_t[j]$ is still $2d$-distant, and since the other implied sequences do not change,
	the claim also holds after iteration $t$.
	
	That  $\langle L_t[1],\dots,L_t[k_t] \rangle$ is $2d$-distant in $\APXSORT$ also follows by induction: If $L[j]$ changes  
	(thus $L[j']$ does not change for all $j'\ne j$), then by hypothesis and the conditions in lines \ref{firstcomp}, \ref{secondcomp}, and  \ref{twoddist},  
	$\pos(L_{t-1}[j-1],\APXSORT) +2d \le \pos(L_{t}[j],\APXSORT)  < \pos(L_{t-1}[j],\APXSORT) \le \pos(L_{t-1}[j+1],\APXSORT)-2d$
	 (for all those entries that exist).
\end{proof}

\begin{lemma}\label{lem:alg2}
	Let $S'=\langle s'_1,\dots,s'_m \rangle$ be the sequence that our Approximation-Algorithm returns. Then, $S'$ is a longest subsequence  of $S$  that is $2d$-distant in $\APXSORT$.
\end{lemma}
\begin{proof}
Lemma~\ref{lem:alg1} implies that $S'$ is a subsequence of $S$ and $2d$-distant in $\APXSORT$. 
Let $S^*=\langle s^*_1,\dots,s^*_{m^*} \rangle$ be a longest subsequence of $S$  that is $2d$-distant in $\APXSORT$. 
We now show that $|S'|\ge |S^*|$.
In particular, we show by induction that after iteration $t^*_i$,
  $\pos(L_{t^*_i}[i],\APXSORT) \le \pos(s^*_i,\APXSORT)$.
For the base case, consider iteration $t^*_1$, where $s^*_1$ is processed. Either $s^*_1$ gets inserted into some position $j\ge 1$, i.e., $L_{t^*_1}[j] = s^*_1$, or not.
If it gets inserted, then by conditions in lines \ref{firstcomp} or \ref{secondcomp} in Algorithm~\ref{alg:apx}, $\pos(L_{t^*_1}[1],\APXSORT) \le \pos(s^*_1,\APXSORT)$. 
	If not, then it must hold that $L_{t^*_1}[1]=L_{t^*_1-1}[1]$ and thus $ \pos(L_{t^*_1}[1],\APXSORT) < \pos(s^*_1,\APXSORT)\, .$
	
For the step case, consider iteration $t_{i+1}^*$, where $s_{i+1}^*$ is processed, and observe that the value of $k$ only increases during the algorithm, and for any $t'<t$ and $j\le k_{t'}$ it holds that 
$\pos(L_{t'}[j],\APXSORT) \ge \pos(L_t[j],\APXSORT)\, . $
%
Therefore, and by induction hypothesis and the assumption that $S^*$ is $2d$-distant in $\APXSORT$, 
 $\pos(L_{t_{i+1}^*-1}[i],\APXSORT)+2d \le \pos(s^*_{i+1},\APXSORT)\, .$ 
And Lemma~\ref{lem:alg1} implies, $\pos(L_{t_{i+1}^*-1}[i],\APXSORT)+2d \le \pos(L_{t_{i+1}^*-1}[i+1],\APXSORT)\,.$
Thus, if $s_{i+1}^*$ does not get inserted, it is because $\pos(L_{t_{i+1}^*-1}[i+1],\APXSORT) < \pos(s^*_{i+1},\APXSORT)$, and if it gets inserted, it will be in some position $j\ge i+1$. In any case, the hypothesis also holds after the iteration iteration ${t_{i+1}^*}$, which means that $S'$ has indeed maximum length.	
\end{proof}

\subsection{Proof of Theorem~\ref{thm:upper}}

We now prove the initially stated Theorem~\ref{thm:upper}, which for convenience, we restate here:
\setcounter{theorem}{0}
\begin{theorem}[Upper Bounds]
	For any sequence $S$ that contains $n$ distinct elements, our {Approximation-Algorithm} computes an $O(\log n)$-approximation of the longest increasing sequence of $S$, in $O(n\log n)$ time, with probability at least $1-\frac{1}{n}$.
\end{theorem}
\begin{proof}
	Let $d\in O(\log n)$ according to Theorem~\ref{thm:sort}, such that with probability $1-\frac{1}{n}$, the maximum dislocation in $\APXSORT$ is at most $d$. If this is true, by  Lemmata~\ref{lem:relative-order}-\ref{lem:alg2}, our Approximation-Algorithm returns a subsequence $S'$ of $S$ that is increasing, and that has length at least $\frac{LIS(S)}{2d}\in \Omega\big(\frac{LIS(S)}{\log n}\big)$.
	
	The running time consists of the initial sorting, which by Theorem~\ref{thm:sort} takes $O(n\log n)$ time\footnote{By modifying this algorithm so that it returns also the mapping from each element in $S$ to its position in $\APXSORT$ we can obtain the new comparison results in the same time.}, and the $n$ iterations of the algorithm, which take $O(\log n)$ time each if binary search is used to implement line 10. 
	The final construction of the output takes $O(k)$ time, where $k\le LIS(S)\le n$ is the length of the approximation.
\end{proof}

\section{Lower Bound on the Approximation Factor}\label{sec:lowerbound}

We continue this paper with a lower bound on the approximation factor, that implies that the upper bound we showed in Theorem~\ref{thm:upper} is tight up to constant factors. In particular, we prove Theorem~\ref{thm:lower}, which we restate here:

\setcounter{theorem}{1}
\begin{theorem}[Lower Bound -- Approximation Factor]
	There exists a collection of sequences $\mathcal{S}$ (permutations of length $n$) and a probability distribution on $\mathcal{S}$, such that no algorithm can return an $O(\log n)$-approximation (for some suitable hidden constant that depends on $p$) of the longest increasing subsequence with probability $1-\frac{1}{n}$.
\end{theorem}

\setcounter{theorem}{8}

Our proof can be seen as a generalization of the lower bound on the maximum dislocation for sorting (see proof of Theorem 9 in \cite{ISAAC}), where it is shown that two elements whose ranks differ by less than $O(\log n)$ are likely to be indistinguishable by any algorithm, and hence to appear in the wrong relative order. 
Intuitively, the argument there is as follows: consider the sorted sequence and the sequence obtained by swapping two elements, and assume that the comparison outcomes on these sequences look identically. It turns out that the probability of this happening is larger than $\frac{1}{n}$, whenever the rank difference is smaller than $O(\log n)$, since only a small number of comparison outcomes must differ.

This is not enough in our case, since an algorithm could simply ignore such two elements. 
For instance, consider an increasing sequence of $c$ adjacent elements. If the first and the last element are swapped, the algorithm could simply return the subsequence without these two elements and be almost optimal.
A first idea to fix this problem could be to consider the case, where one observes the whole increasing sequence to be reversed. However, to have this happen with probability larger than $\frac{1}{n}$, $c$ needs to be smaller than $O(\sqrt{\log n})$, thus implying a weaker lower bound.

Instead, we shall use a collection of similar sequences (more than two),  such that if an algorithm \emph{succeeds} on one of these sequences  it must fail on another one.

\begin{proof}
	We say that an algorithm \emph{succeeds} if it returns a $(c\log n)$-approximation for any constant $c<\frac{1}{2\log \frac{1-p}{p}}$, otherwise we say it \emph{fails}.
We shall first define our collection $\mathcal{S}$ of similar sequences.
	Let $\eta:=\lceil\frac{\log n}{2\log \frac{1-p}{p}}\rceil$.  Let $S^*$ denote the sequence, in which the largest $\eta$ elements appear first in increasing order and then the remaining elements appear in decreasing order,
	\[S^* := \langle n-\eta+1, \dots, n-1,\boldsymbol{n},~~n-\eta,\dots,1\rangle \, .\]
	Furthermore, for $1\le i<\eta$, let $S_{(i)}$ be the sequence obtained from $S^*$ when the largest element is moved to position $i$, 
	\[S_{(i)} := \langle n-\eta+1,  \dots, n-\eta+(i-1),\boldsymbol{n},n-\eta+i,\dots, n-1,~~n-\eta,\dots,1\rangle\, .\]
	%
	Now, let $\mathcal{S}:=\{S^*,S_{(1)},S_{(2)},\dots,S_{(\eta-1)}\}$ (note that basically $S^*=S_{(\eta)}$) and let $\mathcal{P}$ be the uniform distribution over $\mathcal{S}$. We will show (proof by contradiction) that no algorithm succeeds on this pair ($\mathcal{S},\mathcal{P}$) with probability at least $1-\frac{1}{n}$.
	
	Assume towards a contradiction that algorithm $A$ succeeds with high probability on a sequence $S'$ chosen uniformly at random from $\mathcal{S}$, 
	i.e., \[\Pr(A(S') \text{ succeeds}) = \sum_{i=1}^{\eta}\Pr(A(S_{(i)})\text{ succeeds})\cdot \Pr(S'=S_{(i)}) 
	\ge 1-\frac{1}{n}\, .\] 
	This implies that 
	\begin{equation}\label{eq:1}
	P:=\Pr(A(S^*)\text{ succeeds}) \ge 1-\frac{\eta}{n}\, ,
	\end{equation}
	since by hypothesis and assuming the case where the algorithm succeeds on all the other input sequences (i.e., best case for the algorithm, worst case for the proof),  $\frac{P}{\eta}+\frac{\eta-1}{\eta} \ge 1-\frac{1}{n}$ resolves to \eqref{eq:1}.
	
	Let $C\in \{\prec,\succ\}^{\binom{n}{2}}$, then $A(S,C)$ means that algorithm $A$ runs on sequence $S$ and observes comparison outcomes $C$.
	Now, consider the set of all comparison outcomes that the algorithm can observe and let 
	$ \mathcal{C}:= \{C\in\{\prec,\succ\}^{\binom{n}{2}} \colon A(S^*,C) \text{ succeeds} \} $
	denote the set of all possible comparison outcomes for which $A$ succeeds on input $S^*$.
	We define $R(S)\in\{\prec,\succ\}^{\binom{n}{2}}$ to be the random variable
	corresponding to the comparison outcomes as they would be observed by the algorithm when the input sequence is $S$. Then, the probability that $A(S^*)$ succeeds is expressed by the total probabilities of the events that $A$ observes comparison outcomes in $\mathcal{C}$,
	\begin{equation}\label{eq:2}
	P=\Pr(A(S^*)\text{ succeeds}) = \sum_{C\in\mathcal{C}}\Pr(R(S^*)=C)\, .
	\end{equation}
	
	Before we continue the proof, we shall first show the following lemma.
	\begin{lemma}
		$\forall S\in \mathcal{S}\setminus\{S^*\}$ and $C\in \{\prec,\succ\}^{\binom{n}{2}}$, $\Pr(R(S)=C) > \Pr(R(S^*)=C)\cdot\left(\frac{p}{1-p}\right)^\eta\,.$
	\end{lemma}
	\begin{proof}
		Consider $S^*=\langle s^*_1,\dots,s^*_n \rangle$ and $C$ and let $E(S^*,C)$ be the set of wrong comparison results, i.e., the set of pairs ($s^*_i,s^*_j$) with $i<j$ such that either $s^*_i<s^*_j$ and $c_{(i-1)n+j}=``\succ"$ (i.e., $s^*_i\succ s^*_j$) or $s^*_i>s^*_j$ and $c_{(i-1)n+j}=``\prec"$.
		Thus,
		\[\Pr(R(S^*)=C) = (1-p)^{{\binom{n}{2}}-|E(S^*,C)|}\cdot p^{|E(S^*,C)|} = (1-p)^{{\binom{n}{2}}}\cdot \left(\frac{p}{1-p}\right)^{|E(S^*,C)|}\, .
		 \]
		Now consider $S=S_{(k)}=\langle s_1,\dots,s_n \rangle$ and observe that only the relative order of the pairs $(s_k,s_j)$ with $k<j\le \eta$, changed compared to $S^*$. This implies that there can be at most $\eta-k<\eta$ additional wrong comparison results, i.e., $|E(S,C)|< |E(S^*,C)| + \eta$.
		Therefore, and since $\frac{p}{1-p} \le 1$,
		\begin{align*}
			\Pr(R(S)=C) &=  (1-p)^{{\binom{n}{2}}}\cdot \left(\frac{p}{1-p}\right)^{|E(S,C)|} \\&> (1-p)^{{\binom{n}{2}}}\cdot \left(\frac{p}{1-p}\right)^{|E(S^*,C)|+\eta} = \Pr(R(S^*)=C)\cdot\left(\frac{p}{1-p}\right)^\eta\, .  \tag*{\qedhere}
		\end{align*}
	\end{proof}

	\noindent
	\textit{Continuation of the Proof of Theorem 2. }
	Now notice that in order to succeed, $A$ needs to return at least two of the first $\eta$ elements in $S^*$.	
	Therefore, we can  map every $C\in \mathcal{C}$ to a (not necessarily unique) sequence of $\mathcal{S}$ as follows: for each $C\in\mathcal{C}$, let $i_C$ be the position of the first element that $A(S^*,C)$ returns and let $S(C) := S_{(i_C)}$. (Note that $i_C<\eta$ as otherwise $A$ does not return at least two elements of the first $\eta$ elements in $S^*$.)
	For each $S\in \mathcal{S}\setminus\{S^*\}$,
	\begin{align*}
		\Pr(A(S) \text{ fails}) &\ge \sum_{C\in\mathcal{C}\colon S=S(C)}\Pr(R(S)=C)> \sum_{C\in\mathcal{C}\colon S=S(C)}\Pr(R(S^*)=C)\cdot\left(\frac{p}{1-p}\right)^\eta
	\end{align*}
And as a consequence, for $S'\in \mathcal{S}$ chosen uniformly at random, 
	\begin{align*}
		\Pr(A(S') \text{ fails}) 
		&\ge \sum_{S\in \mathcal{S}\setminus\{S^*\}}\Pr(S'=S)\cdot\Pr(A(S) \text{ fails})\\
		&> \sum_{S\in \mathcal{S}\setminus\{S^*\}}\frac{1}{\eta } \sum_{C\in\mathcal{C}\colon S=S(C)}\Pr(R(S^*)=C)\cdot\left(\frac{p}{1-p}\right)^\eta\\
		&= \frac{1}{\eta } \left(\frac{p}{1-p}\right)^\eta \sum_{S\in \mathcal{S}\setminus\{S^*\}} \sum_{C\in\mathcal{C}\colon S=S(C)}\Pr(R(S^*)=C)\\
		&\ge \frac{1}{\eta } \left(\frac{p}{1-p}\right)^\eta \sum_{C\in\mathcal{C}}\Pr(R(S^*)=C)
		\ge \frac{1}{\eta } \left(\frac{p}{1-p}\right)^\eta \left(1-\frac{\eta}{n}\right)\,,
	\end{align*}
where from line 3 to line 4 we use that every instance of comparison results is mapped to exactly one sequence, and on the last line we use Equations \eqref{eq:1} and \eqref{eq:2}.
Now, observe that for $n$ large enough, $\left(1-\frac{\eta}{n}\right)> \frac{1}{2}$ and that, by our choice of $\eta$, $\left(\frac{p}{1-p}\right)^\eta \ge \frac{1}{\sqrt{n}}$. Therefore,
\[	\Pr(A(S') \text{ fails}) > \frac{2\log \frac{1-p}{p}}{\log n}\cdot \frac{1}{\sqrt{n}}\cdot \frac{1}{2} > \frac{1}{n}\, .\]
However, this contradicts our assumption that $A$ succeeds with high probability.
\end{proof}

The lower bound shown in Theorem~\ref{thm:lower} holds for all deterministic algorithms, but can be expanded to also hold for probabilistic algorithms as explained in the following remark.

\begin{remark}
	To make the lower bound on the approximation factor work also for any randomized algorithm $A$, we can turn $A$ into a deterministic version by fixing a sequence $\lambda\in\{0,1\}^t$ random bits that can be used by the algorithm. Thus, for the resulting deterministic algorithm $A_\lambda$, the lower bound holds. Let $p_\lambda$ be the probability to generate the sequence $\lambda$ of random bits. To lower bound the probability that $A(S')$ fails, where $S'$ is chosen uniformly at random from $\mathcal{S}$, one simply needs to sum over all $\lambda$ the probabilities that $A_\lambda$ fails multiplied by $p_\lambda$, i.e.,
	$\Pr(A(S') \text{ fails}) = \sum_{\lambda\in \{0,1\}^t}\Pr(A_\lambda(S') \text{ fails})\cdot p_\lambda \ge \frac{1}{n}\sum_{\lambda\in \{0,1\}^t} p_\lambda=\frac{1}{n}\, .$
\end{remark}

\section{Lower Bound on the Running Time}\label{sec:lowerbound-time}
We complement this paper by showing that the running time of our Approximation-Algorithm is asymptotically optimal. 
In \cite{Fredman75}, it is shown that (in the error-free model) computing the longest increasing subsequence is at least as hard as sorting. 
We will use this proof to informally show Theorem~\ref{thm:lower-time} which we restate here (we postpone a formal proof to the full version of the paper):
\setcounter{theorem}{2}
\begin{theorem}[Lower Bound -- Running Time] 
	Any $\log n$-approximation algorithm for longest increasing subsequence requires $\Omega(n \log n)$ comparisons, even if no  errors occur.
\end{theorem}

The proof techniques of the lower bound in \cite{Fredman75} are as follows: Assume that we are in the error-free case.
Consider the easier problem of deciding on a given sequence $S$ of $n$ distinct elements whether $|LIS(S)| < k$, and consider the comparison tree of an algorithm $A$ with leaves that tell as an answer to this question either ``yes'' or ``no''. Without loss of generality, assume that no useless comparisons are made on a root to a leaf path (i.e., no comparison twice and no comparisons whose outcome is predictable by the outcomes of previous comparisons).

Every leaf $\ell$ can be associated with a partial order implied by a set of linear orderings on $S$ that are consistent with the transitive closure of the comparisons performed on the path from the root to $\ell$.
If the answer in a leaf is ``yes'', this implies that there are no $k$ elements of $S$ that are pairwise incomparable in this partial order (i.e., the relative order of every pair is neither tested in any comparison on the path, nor implied by other comparisons), as otherwise, these elements could possibly form an increasing sequence of length $k$.
Such a subset of elements is called  \emph{antichain}, while a \emph{chain} is a subset of elements that are linearly ordered. 
An important property of chains and antichains used in the proof is, that in a ``yes''-leaf, the elements can be partitioned into less than $k$ chains, since in any partial order, the elements can be partitioned into $m$ chains, where $m$ is the size of the largest antichain. Furthermore, given such a partition into (less than) $k$ chains, the elements can be sorted with $n\log k +O(n)$ comparisons (think for instance of natural merge sort).

In order to lower bound the number of comparisons needed to end in a ``yes''-leaf, algorithm $A$ can be extended to $A^*$ as follows: whenever $A$ concludes to be in a ``yes''-leaf, $A^*$ continues to completely sort the elements of $S$ (which requires no more than $n\log k +O(n)$ further comparisons).
Let $S(n,k)$ denote the number of linear orderings of the elements in $S$ that end in a ``yes''-leaf. 
Then, $
S(n,k) \ge n!(1-{{\binom{n}{k+1}}}/{(k)!})\,,
$ since there are $n!$ different linear orderings and $\binom{n}{k}$ possible subsequences of size $k$ each increasing with probability $1/k!$.
The comparison tree corresponding to $A^*$ has thus at least $S(n,k)$ leaves, and therefore must perform at least $\log S(n,k)$ comparisons in its worst case. Therefore, $A$ must perform at least $\log S(n,k) - n\log k -O(n)$ comparisons in its worst case to end up in a ``yes''-leaf, which is $\Omega(n\log n)$ when choosing $k=3\cdot \sqrt{n}$.

We can use the above proof techniques to show that every algorithm, that computes a $\log n$-approximation on longest increasing subsequence must perform at least $\Omega(n\log n)$ comparisons.
Let $ B$ be an $\log n$-approximation algorithm for $LIS(S)$ under our error model (i.e., we can always simulate our error model in the error-free case) and consider a relaxation of the problem of determining whether $|LIS(S)|$ is smaller than $k \log n$. 
In this relaxation we require the answer to be ``yes'' (resp. ``no'') if $|LIS(S)|< k$ (resp. $|LIS(S)|\ge k\log n$), while we do not impose any restriction on the range $k\le|LIS(S)|<k \log n$.
It is clear that algorithm $B$ can be used to solve this relaxed problem without increasing the number of needed comparisons.
Therefore,  the associated comparison tree must reach a leaf corresponding to answer ``yes'' for all linear orderings on the elements in $S$ that contain no increasing subsequence of length $k$, while the largest antichain in any such an ordering is smaller than $k \log n$. 
This implies that $B^*$ (still in the error-free case) needs at least $n\log(k\log n) +O(n)$ further comparisons in the worst case to sort the elements in $S$, and $B$ needs at least $\log S(n,k) - n\log(k\log n) -O(n)$ comparisons in the worst case to end in a ``yes''-leaf, which is in $\Omega(n\log n)$ if we set $k=n^{2/3}$. 

Finally, we can conclude that our Approximation-Algorithm performs in asymptotically optimal time, since we can always simulate our error model in the error-free case.

\section{Conclusion}\label{sec:conclusion}
Although a logarithmic approximation ratio might not seem very exciting at first glance, it turns out that this is the best one that can be obtained in the presence of persistent comparison errors. In this respect, it is interesting to see that there exist such a simple recipe to compute a logarithmic approximation. A recipe that can use as a black box any algorithm that computes a longest increasing sequence if no comparison errors happen:
\begin{itemize}
	\item First, obtain an approximately sorted sequence $\APXSORT$ of the elements such that the maximum dislocation is $d$ and redefine the comparisons according to this order. 
	Then, partition the elements into $2d$ subsets, such that every $2d$-th element in $\APXSORT$ gets into the same partition, and obtain $2d$ input subsequences based on this partition.
	Finally, run the algorithm on every input subsequence and return the longest result.
\end{itemize}
As indicated earlier, our Approximation-Algorithm has the advantage, that it performs much better than $O(\log n)$-approximate on many input sequences and is even optimal in the case where the longest increasing subsequence is already $2d$-distant in $\APXSORT$, whereas this is not necessarily true when using the simple recipe. Moreover, it is easy to observe that the Approximation-Algorithm is never worse than the recipe.

\smallskip
Finally, we would like to explain how the upper bound on the approximation factor can be generalized. Our Approximation-Algorithm actually succeeds whenever the approximately sorted sequence has maximum dislocation at most $d$. This implies that the result can be parametrized and also used in other models with comparison comparison errors.
\begin{itemize}
	\item Whenever one can obtain a total order with maximum dislocation $d$, the Approximation-Algorithm is $2d$-approximative.
\end{itemize}
Consider for instance the so-called \emph{threshold}-model \cite{ajtai2016sorting,FunkeMN05,GeissmannP18}, where comparisons between numbers that differ by more than some threshold $\tau$ are always correct, while those between numbers that differ by less than $\tau$ can fail persistently (with some probability possibly depending on the difference or even adversarially). If the input sequence $S$ is a permutation of the numbers $\{1\dots,n\}$, running Quicksort in this error model yields a sequence with maximum dislocation $2\tau$ (see \cite{GeissmannP18}). Thus, our Approximation-Algorithm finds a $4\tau$-approximation of the longest increasing subsequence in $S$.

\bibliography{references}

\begin{thebibliography}{10}

\bibitem{ajtai2016sorting}
Mikl{\'o}s Ajtai, Vitaly Feldman, Avinatan Hassidim, and Jelani Nelson.
\newblock Sorting and selection with imprecise comparisons.
\newblock {\em ACM Transactions on Algorithms}, 12(2):19, 2016.

\bibitem{aldous1999longest}
David Aldous and Persi Diaconis.
\newblock Longest increasing subsequences: from patience sorting to the
  baik-deift-johansson theorem.
\newblock {\em Bulletin of the American Mathematical Society}, 36(4):413--432,
  1999.

\bibitem{0305-4470-39-29-L01}
Eitan Bachmat, Daniel Berend, Luba Sapir, Steven Skiena, and Natan Stolyarov.
\newblock Analysis of aeroplane boarding via spacetime geometry and random
  matrix theory.
\newblock {\em Journal of Physics A: Mathematical and General}, 39(29):L453,
  2006.

\bibitem{baik1999distribution}
Jinho Baik, Percy Deift, and Kurt Johansson.
\newblock On the distribution of the length of the longest increasing
  subsequence of random permutations.
\newblock {\em Journal of the American Mathematical Society}, 12(4):1119--1178,
  1999.

\bibitem{BespamyatnikhS00a}
Sergei Bespamyatnikh and Michael Segal.
\newblock Enumerating longest increasing subsequences and patience sorting.
\newblock {\em Inf. Process. Lett.}, 76(1-2):7--11, 2000.

\bibitem{BravermanM08}
Mark Braverman and Elchanan Mossel.
\newblock Noisy sorting without resampling.
\newblock In {\em Proceedings of the Nineteenth Annual {ACM-SIAM} Symposium on
  Discrete Algorithms, {SODA} 2008, San Francisco, California, USA, January
  20-22, 2008}, pages 268--276, 2008.

\bibitem{ChandramouliG14}
Badrish Chandramouli and Jonathan Goldstein.
\newblock Patience is a virtue: revisiting merge and sort on modern processors.
\newblock In {\em International Conference on Management of Data, {SIGMOD}
  2014, Snowbird, UT, USA, June 22-27, 2014}, pages 731--742, 2014.

\bibitem{CrochemoreP10}
Maxime Crochemore and Ely Porat.
\newblock Fast computation of a longest increasing subsequence and application.
\newblock {\em Inf. Comput.}, 208(9):1054--1059, 2010.

\bibitem{delcher}
Arthur~L. Delcher, Simon Kasif, Robert~D. Fleischmann, Jeremy Peterson, Owen
  White, and Steven~L. Salzberg.
\newblock Alignment of whole genomes.
\newblock {\em Nucleic Acids Research}, 27(11):2369--2376, 1999.

\bibitem{Fredman75}
Michael~L. Fredman.
\newblock On computing the length of longest increasing subsequences.
\newblock {\em Discrete Mathematics}, 11(1):29--35, 1975.

\bibitem{FunkeMN05}
Stefan Funke, Kurt Mehlhorn, and Stefan N{\"{a}}her.
\newblock Structural filtering: a paradigm for efficient and exact geometric
  programs.
\newblock {\em Comput. Geom.}, 31(3):179--194, 2005.

\bibitem{ISAAC}
Barbara Geissmann, Stefano Leucci, Chih{-}Hung Liu, and Paolo Penna.
\newblock Sorting with recurrent comparison errors.
\newblock In {\em 28th International Symposium on Algorithms and Computation,
  {ISAAC} 2017, December 9-12, 2017, Phuket, Thailand}, pages 38:1--38:12,
  2017.

\bibitem{STACS}
Barbara Geissmann, Stefano Leucci, Chih{-}Hung Liu, and Paolo Penna.
\newblock Optimal dislocation with persistent errors in subquadratic time.
\newblock In {\em 35th Symposium on Theoretical Aspects of Computer Science,
  {STACS} 2018, February 28 to March 3, 2018, Caen, France}, pages 36:1--36:13,
  2018.

\bibitem{OPT-SORTING}
Barbara Geissmann, Stefano Leucci, Chih{-}Hung Liu, and Paolo Penna.
\newblock {Optimal Sorting with Persistent Comparison Errors}.
\newblock {\em ArXiv e-prints}, April 2018.

\bibitem{GeissmannP18}
Barbara Geissmann and Paolo Penna.
\newblock Inversions from sorting with distance-based errors.
\newblock In {\em {SOFSEM} 2018: Theory and Practice of Computer Science - 44th
  International Conference on Current Trends in Theory and Practice of Computer
  Science, Krems, Austria, January 29 - February 2, 2018, Proceedings}, pages
  508--522, 2018.

\bibitem{KleinPSW11}
Rolf Klein, Rainer Penninger, Christian Sohler, and David~P. Woodruff.
\newblock Tolerant algorithms.
\newblock In {\em Algorithms - {ESA} 2011 - 19th Annual European Symposium,
  Saarbr{\"{u}}cken, Germany, September 5-9, 2011. Proceedings}, pages
  736--747, 2011.

\bibitem{MASEK198018}
William~J. Masek and Michael~S. Paterson.
\newblock A faster algorithm computing string edit distances.
\newblock {\em Journal of Computer and System Sciences}, 20(1):18 -- 31, 1980.

\bibitem{potts1991permutation}
Chris~N. Potts, David~B. Shmoys, and David~P. Williamson.
\newblock Permutation vs. non-permutation flow shop schedules.
\newblock {\em Operations Research Letters}, 10(5):281--284, 1991.

\bibitem{YangHC05}
I{-}Hsuan Yang, Chien{-}Pin Huang, and Kun{-}Mao Chao.
\newblock A fast algorithm for computing a longest common increasing
  subsequence.
\newblock {\em Inf. Process. Lett.}, 93(5):249--253, 2005.

\bibitem{Zhang03}
Hongyu Zhang.
\newblock Alignment of {BLAST} high-scoring segment pairs based on the longest
  increasing subsequence algorithm.
\newblock {\em Bioinformatics}, 19(11):1391--1396, 2003.

\end{thebibliography}

\end{document}